\documentclass[11pt,a4paper,reqno]{amsart}
\usepackage{amsmath,amsthm,amsfonts,dsfont,color}
\newcommand{\beq}{\begin{equation}}
\newcommand{\eeq}{\end{equation}}
\newcommand{\ba}{\begin{alinged}}
\newcommand{\ea}{\begin{aligned}}
\newcommand{\ud}{\mathrm{d}}

\newcommand{\cD}{{\mathcal D}}

\newcommand{\rz}{{\mathbb R}}
\newcommand{\nz}{{\mathbb N}}

\DeclareMathOperator{\supp}{supp}

\DeclareMathOperator{\ran}{ran}

\newcommand{\eins}{\mathds{1}}
\newtheorem{theorem}{Theorem}[section]

\newtheorem{prop}[theorem]{Proposition}
\newtheorem{cor}[theorem]{Corollary}
\theoremstyle{definition}

\newtheorem{remark}[theorem]{Remark}
\begin{document}
\title[Bound states of a pair of particles]{Bound states of~a~pair of~particles
on~the~half-line with~a~general interaction potential}
\author[Egger]{Sebastian Egger}
\address[S. Egger]{Department of Mathematics, Technion-Israel Institute of Technology, 629 Amado Building, Haifa 32000, Israel}
\email{egger@technion.ac.il}
\author[Kerner]{Joachim Kerner}
\address[J. Kerner]{Department of Mathematics and Computer Science, FernUniversit\"{a}t in Hagen, 58084 Hagen, Germany}
\email{Joachim.Kerner@fernuni-hagen.de}
\author[Pankrashkin]{Konstantin Pankrashkin}%
\address[K. Pankrashkin]{Laboratoire de Math\'ematiques d'Orsay, Univ.~Paris-Sud, CNRS, Universit\'e Paris-Saclay, 91405 Orsay, France}
\email{konstantin.pankrashkin@math.u-psud.fr}
\begin{abstract} In this paper we study an interacting two-particle system on the positive half-line $\mathbb{R}_+$. We focus on spectral properties of the Hamiltonian for a large class of two-particle potentials. We characterize the essential spectrum and prove, as a main result, the existence of eigenvalues below the bottom of it. We also prove that the discrete spectrum contains only finitely many eigenvalues.
\end{abstract}
\maketitle

%
\section{Introduction}

In this paper we are concerned with spectral properties of an interacting two-particle system moving on the half-line $\mathbb{R}_+:=(0,\infty)$. More specifically, we consider the (two-particle) Hamiltonian
in $L^2(\rz_+\times\rz_+)$ given by
\begin{equation}\label{HamiltonianIntro}
H=-\frac{\partial^2}{\partial x_1^2}-\frac{\partial^2}{\partial x_2^2}+v\left(\dfrac{|x_1-x_2|}{\sqrt{2}}\right)\, ,
\end{equation} 
with an interaction potential $v:\rz_+ \rightarrow \mathbb{R}$ belonging to a large class covering all physically meaningful potentials including, e.g., quadratic and Lennard-Jones-type potentials. Note that the factor $\tfrac{1}{\sqrt{2}}$ in the argument of $v$ is only chosen for further convenience.
Very informally, our main result is that if the potential $v$ creates a bound state for the respective one-dimensional Schr\"odinger operator
on the half-line, then it creates at least one eigenvalue of $H$ with a strictly lower energy.

The present work is a far-reaching extension of the previous work~\cite{KernerMuhlenbruch2} in which a similar result was obtained for a specific class of hard wall potentials $v$.
As described in~\cite{KernerElectronPairs,KernerInteractingPairs}, the presence of a discrete spectrum leads to a (Bose-Einstein) condensation of pairs in a gas of bosonic, non-interacting pairs with each pair described by~\eqref{HamiltonianIntro}. A condensation of pairs of electrons, on the other hand, is the key mechanism in the formation of the superconducting phase in type-I superconductors~\cite{CooperBoundElectron,BCSI}. Hence, the extension of the model discussed in this paper is expected to have also interesting applications in solid-state physics.
One should emphasize on the fact that only very few two-particle problems admit an explicit solution, see e.g. \cite{busch}, so qualitative results are of a particular importance.

%


Let us introduce some notions used throughout the paper: To keep the notation as simple as possible, we will work with real-valued Hilbert spaces.
For a self-adjoint and semi-bounded operator $A$ we denote by $\cD(A)$ its domain and by
$\cD[A]$ the domain of the associated bilinear form (which will often referred to as the form domain of $A$). The bilinear form itself will be denoted
as $A[\cdot,\cdot]$, the spectrum and the essential spectrum of $A$ will be denoted by $\sigma(A)$
and $\sigma_\text{ess}(A)$ respectively.

Let $v$ be a real-valued potential on $\rz_+$ with the following properties:
\begin{enumerate}
	\item[(A)] $v \in L^{1}_\text{loc}(\mathbb{R}_+)$ and $\max\{-v,0\}\in L^{\infty}(\mathbb{R}_+)$,
	\item[(B)] The one-particle Schr\"odinger operator
	\[
	h:=-\dfrac{\ud^2}{\ud x^2}+v(x)
	\] in $L^2(\rz_+)$, which is rigorously	defined through its form
	\begin{gather*}
	h[\varphi,\varphi]=\int_{\mathbb{R}_+} \big((\varphi'(x)^2+v(x)\varphi(x)^2\big)\, \ud x\,,\\
		\cD[h]=\Big\{\varphi\in H^1(\rz_+): \int_{\rz_+} v(x)\varphi(x)^2\, \ud x<+\infty\Big\}\,,
	\end{gather*}
		is such that the bottom of the spectrum $\inf \sigma(h)=:\varepsilon_0$ is an isolated eigenvalue,
	\item[(C)] The bottom eigenvalue is strictly lower than the values of $v$ at infinity, i.e. it holds $\varepsilon_0 < \liminf_{x \rightarrow \infty} v(x):=v_{\infty}$.
\end{enumerate}
The assumption (C) is to avoid potentials with a pathological behavior, and it holds
for the physically reasonable cases. It is well known that that the assumptions (B) and (C) are satisfied in two important cases:
\begin{itemize}
\item[(a)] for $v_\infty=+\infty$,
\item[(b)] $v_\infty<\infty$ and $v-v_\infty\in L^1(\rz_+)$ with
$\displaystyle\int_{\rz_+}\big( v(x)-v_\infty\big)\ud x <0$
\end{itemize}
(see Propositions~\ref{cond1} and~\ref{cond2} in Appendix).

For potentials $v$ which are sufficiently regular near $0$, for example, for $v|_{(0,1)}\in L^2(0,1)$,
it is standard to see that the above operator $h$ corresponds to the \emph{Neumann}
condition $\varphi'(0)=0$ at the origin.
In general, the operator $h$ can be in the limit point case at $0$ (if $v$ diverges very fast at zero) in which case the characterization of boundary conditions is more involved. However, this subtlety is of no importance for our constructions in the following.

The associated two-particle Schr\"odinger operator
\[
H=-\Delta + v\left(\dfrac{|x_1-x_2|}{\sqrt{2}}\right)
\]
in $L^2(\rz_+^2)$ is rigorously defined through its form,
\begin{gather*}
H[\varphi,\varphi]=\iint_{\rz_+^2} \Big(\big|\nabla \varphi(x_1,x_2)\big|^2+v\left(\dfrac{|x_1-x_2|}{\sqrt{2}}\right) \varphi(x_1,x_2)^2\Big)\, \ud x_1\ud x_2\,,\\
\cD[H]=\big\{\varphi\in H^1(\rz_+^2):\, H[\varphi,\varphi]<\infty\big\}\,;
\end{gather*}
note the factor 
$\tfrac{1}{\sqrt{2}}$ in the argument of $v$ which is chosen for convenience in order to have less factors in later computations.
Our results are summarized as follows:
\begin{theorem}\label{TheoremMainResult}
The essential spectrum of $H$ is $[\varepsilon_0,+\infty)$, and its discrete spectrum is non-empty and finite.
\end{theorem}
We remark that the presence of a non-empty discrete spectrum is probably the most important result.
It relies on a rather involved construction of a test function whose structure was proposed in \cite{lupan}
for a different problem involving specific potentials with explicitly known ground states, and
it also appeared in e.g.~\cite{hm,P}. So we propose another extension to rather general operators
and hope that it can be used beyond our framework {(See e.g. Remark~\ref{kp1} below.)}
The proof of the finiteness of the discrete spectrum essentially follows the scheme of~\cite{MT}
for another specific operator and essentially represents a realization of the Feshbach
projection method, which was also used in \cite{KP}. A new ingredient is delivered by the fact that
some new properties of the ground state of $h$ should be established first.
The fact that we work with rather singular potentials $v$,
which can be non-integrable near $0$, brings a number
of technical subtleties concerning the regularity of functions, and we collect the respective results
on one-dimensional Schr\"odinger operators in Section~\ref{Appendix}.

\section{Proof of Theorem~\ref{TheoremMainResult}}\label{SectionProof}
\subsection{Reductions by symmetries}
Let us first perform some standard reductions in order to deal with a model case.
Denote
\begin{equation*}
\Omega:=\{(x_1,x_2) \in \mathbb{R}\times\mathbb{R}_+: |x_1|< x_2 \}
\end{equation*}
and consider the diffeomorphism (rotation by $\pi/4$)
\[
\Phi:\Omega\to\rz_+^2\, , \quad
\Phi(x_1,x_2)=\dfrac{1}{\sqrt{2}}\,(x_2+x_1,x_2-x_1)\, ,
\]
and the unitary transform (pull back) $U:L^2(\rz_+^2)\to L^2(\Omega)$,
$U\varphi=\varphi\circ \Phi\,$. Using the standard change of variables one easily checks that
\[
U\cD[H]=\cD[Q]\, , \quad H[\varphi,\varphi]=Q[U\varphi,U\varphi]\, ,
\]
with $Q$ being the operator in $L^2(\Omega)$ given by its form
\begin{gather*}
Q[\varphi,\varphi]=\iint_\Omega \big(|\nabla\varphi(x_1,x_2)|^2+v\big(|x_1|\big)\varphi(x_1,x_2)^2\big)\, \ud x_1\ud x_2\, ,\\
\cD[Q]=\big\{\varphi\in H^1(\Omega):\, Q[\varphi,\varphi]<\infty\big\}\,,
\end{gather*}
which is then unitarily equivalent to $H$. To use the parity with respect to $x_1$
we consider the right half of $\Omega$,
\[
\Omega_0:=\big\{(x_1,x_2) \in \mathbb{R}^2: \ 0< x_1< x_2 \big\},
\]
and the unitary transform
\begin{gather*}
\Theta:L^2(\Omega)
\ni \varphi\mapsto (U_+\varphi,U_-\varphi)\in L^2(\Omega_0)\times L^2(\Omega_0)\, ,\\
U_\pm\varphi(x_1,x_2)=\dfrac{\varphi(x_1,x_2)\pm \varphi(-x_1,x_2)}{\sqrt 2}\,.
\end{gather*}
If one introduces self-adjoint operators $Q_\pm$ in $L^2(\Omega_0)$ given by
\begin{align*}
Q_\pm[\varphi,\varphi]&=\iint_{\Omega_0} \big(|\nabla\varphi(x_1,x_2)|^2+v(x_1)\varphi(x_1,x_2)^2\big)\, \ud x_1\ud x_2\, ,\\
\cD(Q_+)&=\big\{\varphi\in H^1(\Omega_0):\, Q_+[\varphi,\varphi]<\infty\big\}\, ,\\
\cD(Q_-)&=\big\{\varphi\in H^1_0(\Omega_0):\, Q_-[\varphi,\varphi]<\infty\big\}\, ,
\end{align*}
then one easily checks that 
\[
\cD[Q_\pm]:=U_\pm\cD[Q], \qquad
Q[\varphi,\varphi]=Q_+[U_+\varphi,U_+\varphi]+Q_-[U_-\varphi,U_-\varphi].
\]
It follows that $Q$ (hence, also $H$) is unitarily equivalent to $Q_+\oplus Q_-$.
As the bilinear form of $Q_+$ is an extension
of that for $Q_-$, it follows by the min-max principle
that  $\lambda:=\inf\sigma_\text{ess}(Q_+)\le\inf\sigma_\text{ess}(Q_-)$
and that the number of eigenvalues of $Q_-$ below $\lambda$ does not exceed that for  $Q_+$.

Therefore, $\inf\sigma_{\text{ess}}(H)=\min\big\{\inf\sigma_\text{ess}(Q_{-}),\inf\sigma_\text{ess}(Q_+)\big\}=\inf\sigma_\text{ess}(Q_+)$, and the non-emptyness and
finiteness of the discrete spectrum of $Q_+$ will imply the non-emptyness and finiteness of the discrete spectrum of $H$. This shows that Theorem~\ref{TheoremMainResult} becomes a consequence of the following assertion, whose proof will be given in the rest of the section:
\begin{prop}\label{prop1}
The essential spectrum of the operator $Q_+$ is $[\varepsilon_0,+\infty)$, and its discrete spectrum is non-empty and finite.
\end{prop}

{
\begin{remark}\label{kp1} It is clear that the above operators $Q_\pm$ correspond  to the restrictions
of the initial operator $H$  to the symmetric/anti-symmetric
functions, i.e. $\varphi(x_1,x_2)=\pm\varphi(x_1,x_2)$.
While the operator $Q_-$ is ``dominated'' by the operator $Q_+$ (in the sense
that the qualitative spectral picture for $H$ is determined by that of $Q_+$ only), it can be
studied on its own, and the analog of Proposition~\ref{prop1} has then
the following form:
\begin{prop}\label{prop23}
Let $h_0$ be the operator in $L^2(\rz_+)$ with
	\begin{equation}
	 \label{eqh0}
	\begin{split}
	h_0[\varphi,\varphi]&=\int_{\mathbb{R}_+} \big((\varphi'(x)^2+v(x)\varphi(x)^2\big)\, \ud x\,, \\
		\cD[h_0]&=\Big\{\varphi\in H^1_0(\rz_+): \int_{\rz_+} v(x)\varphi(x)^2\, \ud x<+\infty\Big\}. 
		\end{split}
	\end{equation}
	If the bottom of the spectrum $\inf \sigma(h_0)=:\varepsilon_*$ is an isolated eigenvalue
	with $\varepsilon_* < \liminf_{x \rightarrow \infty} v(x):=v_{\infty}$,
then the essential spectrum of $Q_-$ is $[\varepsilon_*,+\infty)$ and the discrete spectrum is
non-empty and finite.
\end{prop}
This can be proved by a literal repetition of the proof of Proposition~\ref{prop1}
given in the following three subsections (see also Remark~\ref{lastrem} in Appendix concerning $h_0$). 
\end{remark}
}

%
\subsection{Essential spectrum}

Let us show the equality $\sigma_\text{ess}(Q_+)=[\varepsilon_0,+\infty)$
by establishing separately the inclusions in both directions.
The constructions of this section are very standard and are given to render a self-contained presentation. 

{In a first step, let us prove first that $\sigma_{\text{ess}}(Q_+)\subset [\varepsilon_0,\infty)$ employing an operator bracketing argument:} For that, we partition $\Omega$ into three subdomains $\Omega_j$, $j=1,2,3$, using the straight lines $x_1=L$ and $x_2=L$ with $L > 0$ large enough. More precisely,
\begin{align*}
\Omega_1&:=\big\{(x_1,x_2)\in\Omega_0:\, x_2<L\big\} && \text{ is the bounded triangle,}\\
\Omega_2&:=\big\{(x_1,x_2)\in\Omega_0:\,  x_1<L,\, x_2>L\big\} && \text{ is the half-infinite strip,}\\
\Omega_3&:=\big\{(x_1,x_2)\in\Omega_0:\, x_1>L\big\} &&\text{ is the remaining infinite sector.}
\end{align*}
Define self-adjoint operators $Q_j$ in $L^2(\Omega_j)$, $j\in\{1,2,3\}$,
through their bilinear forms
\begin{equation*}\begin{split}
Q_j[\varphi,\varphi]=&\iint_{\Omega_j} \big(|\nabla\varphi(x_1,x_2)|^2+v(x_1)\varphi(x_1,x_2)^2\big)\, \ud x_1\ud x_2\,,\\
\cD[Q_j]=&\big\{\varphi \in H^1(\Omega_j): \ Q_j[\varphi,\varphi] < \infty \big\}\,.
\end{split}
\end{equation*} 
Using the canonical orthogonal projections $P_j: L^2(\Omega_0)\rightarrow L^2(\Omega_j)$, defined just as restrictions to $\Omega_j$, we observe that $P_j \cD[Q_0]\subset \cD[Q_j]$
and that $Q_+[\varphi,\varphi]=\sum_{j=1}^3 Q_j[P_j\varphi,P_j\varphi]$ and, in addition, that the map
\[
I:L^2(\Omega_0)\ni \varphi \mapsto
(P_1\varphi, P_2\varphi, P_3\varphi)\in \bigoplus_{j=1}^3 L^2(\Omega_j)
\]
is unitary. It follows by the min-max principle that
\[
\inf\sigma_\text{ess}(Q_0)\ge \inf\sigma_\text{ess}(Q_1 \oplus Q_2 \oplus Q_3)
=\min_{j\in\{1,2,3\}} \inf\sigma_\text{ess} (Q_j)\,.
\]

Since $\Omega_1$ is a bounded Lipschitz domain, the form domain $\cD[Q_1]\subset H^1(\Omega_1)$
is compactly embedded in $L^2(\Omega_1)$, which implies that the spectrum of $Q_1$
is purely discrete. Furthermore, we have $\inf\sigma(Q_3)\ge \inf_{x_1> L} v(x_1)>\varepsilon_0$
for all $L\ge L_0$ with $L_0$ chosen sufficiently large, due to to the assumption (C) on the potential $v$.
It follows that
\begin{equation}
   \label{qqq}
\inf\sigma_\text{ess}(Q_0)\ge \min\big\{ \varepsilon_0,
\inf\sigma_\text{ess}(Q_2)\big\}
\quad \text { for $L\ge L_0$\,.}
\end{equation}

To analyze $Q_2$ we remark first that it admits a separation of variables,
\begin{equation}\label{SeparationVariables}
Q_2=h^{N}_{L} \otimes \eins + \eins \otimes q_2\, ,
\end{equation}
where $h^{N}_{L}$ is the operator in $L^2(0,L)$
associated with the form
\begin{equation*}\begin{split}
h^{N}_{L}[\varphi,\varphi]&=\int_0^L \big((\varphi'(x)^2 + v(x) \varphi(x)^2\big)\, \ud x\, , \\
\cD[h^N_L]&=\big\{\varphi\in H^1(0,L):\, \int_0^L v(x)\varphi(x)^2\, \ud x <\infty\big\}\,,
\end{split}
\end{equation*}
while $q_{2}$ acts in $L^2(L,+\infty)${, being defined via its associated form}
\[
q_2[\varphi,\varphi]=\int_L^\infty \varphi'(x)^2\ud x,
\quad \cD[q_2]=H^1(L,+\infty),
\]
i.e. $q_2$ { acts as} $\varphi\mapsto -\varphi''$ with the Neumann boundary condition at $L$,
and $\sigma(q_2)=\sigma_\text{ess}(q_2)=[0,+\infty)$.
By \eqref{SeparationVariables} there holds $\inf \sigma_{\text{ess}}(Q_2)=\inf \sigma(h^{N}_{L})$. It is standard to see (see Proposition~\ref{CorollaryConvergenceEnergy})
that $\lim_{L\to\infty}\inf\sigma(h^{N}_{L}) =\varepsilon_0$.
By \eqref{qqq} one has $\inf \sigma_{\text{ess}}(Q_+) \geq
\liminf_{L\to+\infty}
\min\big\{ \varepsilon_0,
\inf \sigma(q^{N}_{L})\big\}=\varepsilon_0$.

Now let us show the reverse inclusion $[\varepsilon_0,\infty)\subset\sigma_{\text{ess}}(Q_+)$ by constructing a suitable Weyl sequence.
For that, let $\tau:\mathbb{R} \rightarrow \mathbb{R}$ be a smooth function with $0 \leq \tau\leq 1$ such that
$\tau(x)=1$ for $x \geq 2$ and $\tau(x)=0$ for $x \leq 1$.
Pick any $k\in [0,\infty)$. For  $n\ge 2$ define
$\varphi_n(x_1,x_2)=f_n(x_1) g_n(x_2)$ with
\[
f_n(x_1)=\psi_0(x_1)\tau(n-x_1),\quad
g_n(x_2)=\cos{(kx_2)} \tau(x_2-n)\tau(2n-x_2)\, { .}
\]
{ Then $\varphi_n$} vanishes outside the rectangle $[0,n-1]\times [n+1,2n-1]\subset \Omega_0$, and $\varphi\in \cD(Q_+)$.
For large $n$ one estimates, with a suitable $a>0$,
\[
\|\varphi_n\|^2_{L^2(\Omega_0)}\ge \int_0^{n-2}f_n(x_1)^2{ \ud x_1}\, \int_{n+2}^{2n-2}\cos^2{(kx_2)} \ud x_2\ge an\,.
\]
On the other hand, ${(Q_+ \varphi)}(x_1,x_2)= F_n(x_1)g_n(x_1)+f_n(x_1)G_n$ with
\begin{align*}
F_n(x_1)&= -\psi''_0(x_1) \tau(n-x_1) +2\psi'_0(x_1)\tau'(n-x_1)\\
&\quad - { \psi_0(x_1)}\tau''(n-x_1) +v(x_1)\psi_0(x_1)\tau(n-x_1)\\
&=\varepsilon_0 f_n(x_1) +\Phi_n(x_1)\,,\\
G_n(x_2)&=k^2\cos{(kx_2)} \tau(x_2-n)\tau(2n-x_2)\\
&\quad +2k \sin (kx_2)\big[\tau'(x_2-n)\tau(2n-x_2)-\tau(x_2-n)\tau'(2n-x_2)\big]\\
&\quad -\cos{(kx_2)} \Big[\tau''(x_2-n)\tau(2n-x_2)+2\tau'(x_2-n)\tau'(2n-x_2)\\
&\qquad+\tau(x_2-n)\tau''(2n-x_2)\Big]\\
&=k^2g_n(x_2)+\Psi_n(x_2)\,,
\end{align*}
where with some $b>0$ one has 
\begin{gather*}
\text{$|\Phi_n|\le b \big(|\psi_0|+|\psi'_0|\big)$ with $\supp \Phi_n\subset [n-2,n-1]$}\,,\\
\text{$\|\Psi_n\|_\infty\le b$ with ${ \supp \Psi_n}\subset [n+1,n+2]\,\cup\,[2n-2,2n-1]$}\,.
\end{gather*}
One has $\big(Q_+ -(\varepsilon_0+k^2)\big)\varphi(x_1,x_2)=\Phi_n(x_1)g_n(x_2)+f_n(x_1)\Psi_n(x_2)$ and
\begin{multline*}
\Big\|\big(Q_+ -(\varepsilon_0+k^2)\big)\varphi\Big\|^2_{L^2(\Omega_0)}\le
2 \int_{n-2}^{n-1} \Phi_n(x_1)^2\ud x_1 \int_{n+1}^{2n-1} g_n(x_2)^2\ud x_2\\
+2 \int_{0}^{n-1} f_n(x_1)^2\ud x_1 \bigg(\int_{n+1}^{n+2} + \int_{2n-2}^{2n-1} \Psi_n(x_2)^2\ud x_2\\
\le 4b^2 \int_{n-2}^{n-1} \big(\psi_0^2 + (\psi'_0)^2\big)\ud x_1 \int_{n+1}^{2n-1} \cos (kx_2)^2\ud x_2
+4b^2 \int_{0}^{n-1} \psi_0(x_1)^2\ud x_1\\
{ {\le 4b^2(n-2) \int_{n-2}^{n-1} \big(\psi_0^2 + (\psi'_0)^2\big)\ud x_1}+4b^2.}
\end{multline*}
Therefore, 
\begin{multline*}
\dfrac{\big\|\big(Q_+ -(\varepsilon+k^2)\big)\varphi\big\|^2_{L^2(\Omega_0)}}{\|\varphi_n\|^2_{L^2(\Omega_0)}}\le
\dfrac{{4b^2(n-2)} \displaystyle\int_{n-2}^{n-1} \big(\psi_0^2 + (\psi'_0)^2\big)\ud x_1+4b^2}{an}\\
={\dfrac{4b^2}{a}\Big( 1-\dfrac{2}{n}\Big)}\int_{n-2}^{n-1} \big(\psi_0^2 + (\psi'_0)^2\big)\ud x_1 + \dfrac{4b^2}{a}\,\dfrac{1}{n}\xrightarrow{n\to\infty} 0
\end{multline*}
due to $\psi_0\in H^1(\rz_+)$. Hence, $\varepsilon_0+k^2\in \sigma(Q_+)$ for any $k\ge 0$, in other words, $[\varepsilon_0,\infty)\subset\sigma(Q_+)$. As the set $[\varepsilon_0,\infty)$ has no isolated points, it follows that $[\varepsilon_0,\infty)\subset\sigma_\text{ess}(Q_+)$.
\subsection{Existence of discrete eigenvalues}
In this section we show that the discrete spectrum of $Q_+$ is non-empty.

Recall that the bilinear form of $Q_+$ is given by
\begin{align*}
Q_+[\varphi,\varphi]&=\iint_{\Omega_0} \Big(|\nabla\varphi(x_1,x_2)|^2+v(x_1)\varphi(x_1,x_2)^2\Big)\, \ud x_1\ud x_2\, ,\\
\cD(Q_+)&=\big\{\varphi\in H^1(\Omega_0):\, Q_+[\varphi,\varphi]<\infty\big\}\, .
\end{align*}
As $\inf\sigma_\text{ess}(Q_+)=\varepsilon_0$, it follows by the min-max principle
that the non-emptyness of the discrete spectrum follows
from the existence of a function $\varphi \in \cD[Q_+]$ satisfying the strict inequality $Q_+[\varphi,\varphi] -\varepsilon_0\|\varphi\|^2_{L^2(\Omega_0)}<0$.

We will seek for such a function $\varphi$ in the form $\varphi(x_1,x_2)=\psi_0(x_1)\phi(x_2)\, $,
with $\psi_0$ being as previously the ground state of $h$ and $\phi$ a function to be specified. Due to the standard regularity considerations (see Appendix) there holds
$\psi_0\in C^1(\rz_+)\cap L^\infty(\rz_+)$. With some $\rho>0$ we introduce
\begin{equation}
F(x_2):=\int_{0}^{x_2}\psi_0(x)^2\, \ud x\,, \qquad
\phi(x_2):=F(x_2)^{\rho}.
\end{equation}
It is easily checked (see Proposition~\ref{RegularityFunction}) that $\phi\in H^1(0,a)$ for any $a>0$
provided $\rho>\frac{1}{2}$, which is assumed from now on.
Finally we introduce a smooth cut-off function $\chi$ and the associated truncations $\phi_n$,
$n\in\nz$,
by
\begin{equation*}
0\le \chi\le 1, \quad
\chi(t)=\begin{cases}
0\,,& \quad t \geq 2\,,\\
1\,,& \quad t \leq 1\,,
\end{cases}
\qquad
\phi_n(x_2)=\phi(x_2) \chi\left(\frac{x_2}{n}\right).
\end{equation*}
The function $\varphi$ defined by $\varphi_n(x_1,x_2):=\psi_0(x_1)\phi_n(x_2)$
belongs then to $\cD[Q_+]$ for any $n\in\nz$.
A calculation then yields the following:
\begin{equation*}\begin{split}
Q_+[\varphi_n,\varphi_n]&=\iint_{\Omega_0}\left(|\nabla \varphi_n|^2+v(x_1)|\varphi_n|^2\right) \, \ud x_1\ud x_2 \\
&=\iint_{\Omega_0}\left[\left(\psi_0^{\prime}(x_1)\phi_n(x_2)\right)^2+v(x_1)\psi_0^2(x_1)\phi^2_n(x_2)\right]\, \ud x_1 \ud x_2\\
&\qquad +\iint_{\Omega_0}\left(\psi_0(x_1)\phi_n^{\prime}(x_2)\right)^2 \, \ud x_1 \ud x_2 \\
&=\int_{\mathbb{R}_+}\left(\int_{0}^{x_2} \psi_0^{\prime}(x_1)\psi_0^{\prime}(x_1)+v(x_1)\psi_0(x_1)^2 \, \ud x_1\right)\phi^2_n(x_2)\, \ud x_2 \\
&\qquad+\iint_{\Omega_0}\left(\psi_0(x_1)\phi_n^{\prime}(x_2)\right)^2 \, \ud x_1 \ud x_2\, .
\end{split}
\end{equation*}
An integration by parts ({ which is still possible for singular potentials $v$, see Proposition~\ref{Randterme} in the appendix}) gives
\begin{multline*}
\int_{0}^{x_2} \psi_0^\prime(x_1)\psi_0^{\prime}(x_1)+v(x_1)\psi_0(x_1)^2 \, \ud x_1\\
=\int_{0}^{x_2} \psi_0(x_1)\Big(-\psi_0''(x_1)+v(x_1)\psi_0(x_1)\Big)\ud x_1
+\psi_0(x_2)\psi_0'(x_2),\\
=\varepsilon_0 \int_{0}^{x_2} \psi_0(x_1)^2\, \ud x_1+\psi_0(x_2)\psi_0'(x_2)
\end{multline*}
and { which allows us to write}
\begin{equation}
\label{qpl}\begin{split}
Q_+[\varphi_n,\varphi_n]&=\varepsilon_0 \|\varphi_n\|^2_{L^2(\Omega_0)} +\int_{\mathbb{R}_+}\psi_0(x_2)\psi_0^{\prime}(x_2)\phi^2_n(x_2) \, \ud x_2\\
&\quad +\iint_{\Omega_0}\psi_0(x_1)^2\phi_n^{\prime}(x_2)^2 \, \ud x_1 \ud x_2\, .
\end{split}
\end{equation}
{ Integrating the middle term on the right-hand side by parts one obtains}
\begin{multline*}
\int_{\mathbb{R}_+}\psi_0(x_2)\psi_0^{\prime}(x_2)\phi^2_n(x_2) \, \ud x_2\\
= \dfrac{(\psi_0^2 \phi_n^2)^2(\infty) -(\psi_0^2 \phi_n^2)^2(0)}{2}
-
\int_{\mathbb{R}_+}\psi_0(x_2)^2\phi_n(x_2)\,\phi'_n(x_2) \, \ud x_2\,.
\end{multline*}
One has $\phi_n(0)=\phi_n(\infty)=0$ and $\psi_0\in L^\infty(\rz_+)$, which shows that the first summand on the right-hand side vanishes, and
\[
\int_{\mathbb{R}_+}\psi_0(x_2)\psi_0^{\prime}(x_2)\phi^2_n(x_2) \, \ud x_2\\
= -
\int_{\mathbb{R}_+}\psi_0(x_2)^2\phi_n(x_2)\,\phi'_n(x_2) \, \ud x_2\,.
\]
{ Taking $F'=\psi_0^2$ into account} one rewrites \eqref{qpl} as
\begin{multline*}
Q_+[\varphi_n,\varphi_n]-\varepsilon_0 \|\varphi_n\|^2_{L^2(\Omega_0)}\\
=-\int_{\mathbb{R}_+}\psi^2_0(x_2)\phi_n^{\prime}(x_2)\phi_n(x_2) \, \ud x_2+\int_{\mathbb{R}_+}F(x_2)\phi_n^{\prime}(x_2)^2 \, \ud x_2\\
=\int_{\mathbb{R}_+}\Big(F(x_2)\phi_n^{\prime}(x_2)^2-\phi_n'(x_2)F^{\prime}(x_2)\phi_n(x_2) \Big)\, \ud x_2=:G_n\,.
\end{multline*}
In order to show that the term $G_n$ can be made strictly negative
one uses first the expressions for $F$ and $\phi_n$ to compute
\begin{multline*}
F(x_2)\phi_n^{\prime}(x_2)^2=F(x_2) \bigg(
\rho F(x_2)^{\rho-1}F'(x_2)\chi\Big(\dfrac{x_2}{n}\Big) + \dfrac{1}{n}\, F(x_2)^\rho
\chi'\Big(\dfrac{x_2}{n}\Big)
\bigg)^2\\
=\rho^2 F(x_2)^{2\rho-1}F'(x_2)^2\chi\Big(\dfrac{x_2}{n}\Big)^2
+\dfrac{2\rho}{n}\,F(x_2)^{2\rho}F'(x_2)\chi\Big(\dfrac{x_2}{n}\Big)\chi'\Big(\dfrac{x_2}{n}\Big)\\
+\dfrac{1}{n^2}\, F(x_2)^{2\rho+1}\chi'\Big(\dfrac{x_2}{n}\Big)^2\,,
\end{multline*}
%
and
\begin{multline*}
{\phi_n^{\prime}}(x_2)F^{\prime}(x_2)\phi_n(x_2)\\
=
\bigg(
\rho F(x_2)^{\rho-1}F'(x_2)\chi\Big(\dfrac{x_2}{n}\Big) + \dfrac{1}{n}\, F(x_2)^\rho
\chi'\Big(\dfrac{x_2}{n}\Big)
\bigg) F^{\prime}(x_2) F(x_2)^\rho \chi\Big(\dfrac{x_2}{n}\Big)\\
=\rho F(x_2)^{2\rho-1}F'(x_2)^2\chi\Big(\dfrac{x_2}{n}\Big)^2
+\dfrac{1}{n}\,F(x_2)^{2\rho}F'(x_2)\chi\Big(\dfrac{x_2}{n}\Big)\chi'\Big(\dfrac{x_2}{n}\Big)\,,
\end{multline*}
which yields, for $g_n(x_2):=F(x_2)\phi_n^{\prime}(x_2)^2-{\phi_n^{\prime}}(x_2)F^{\prime}(x_2)\phi_n(x_2)$,
\begin{align*}
g_n(x_2)&=\rho(\rho-1) F(x_2)^{2\rho-1}F'(x_2)^2\chi\Big(\dfrac{x_2}{n}\Big)^2\\
&\quad+\dfrac{2\rho-1}{n}F(x_2)^{2\rho}F'(x_2)\chi\Big(\dfrac{x_2}{n}\Big)\chi'\Big(\dfrac{x_2}{n}\Big)\\
&\quad+\dfrac{1}{n^2}\, F(x_2)^{2\rho+1}\chi'\Big(\dfrac{x_2}{n}\Big)^2\,.
\end{align*}
One then decompse the above term $G_n$ as follows:
\begin{align}
  \label{egn}
G_n&=\int_{\rz_+}g_n(x_2)\ud x_2=\rho(\rho-1)\big(A +B_n\big) +\dfrac{2\rho-1}{n}\,C_n + \dfrac{1}{n^2}\,D_n\,,\\
A&:=\int_{\rz_+}F(x_2)^{2\rho-1}F'(x_2)^2\ud x_2\,,\nonumber\\
B_n&:=\int_{\rz_+}F(x_2)^{2\rho-1}F'(x_2)^2
\bigg(\chi\Big(\dfrac{x_2}{n}\Big)^2-1\bigg)\ud x_2\,,\nonumber\\
C_n&:=\int_{\rz_+}F(x_2)^{2\rho}F'(x_2)\chi\Big(\dfrac{x_2}{n}\Big)\chi'\Big(\dfrac{x_2}{n}\Big)\ud x_2\,,\nonumber\\
D_n&:=\int_{\rz_+}
F(x_2)^{2\rho+1}\chi'\Big(\dfrac{x_2}{n}\Big)^2
\ud x_2\,.\nonumber
\end{align}
We recall that $0\le  F\le 1$ and $F'=\psi_0^2\in L^1(\rz_+)\cap L^\infty(\rz_+)$, which
ensures the finiteness of the integrals. One easily sees that $A>0$,
while $\lim_{n\to+\infty} B_n=0$. We then estimate
\begin{align*}
|C_n|&\le \|\chi'\|_\infty
\int_{\rz_+}F(x_2)^{2\rho}F'(x_2)\ud x_2\\
&=\|\chi'\|_\infty \dfrac{F(\infty)^{2\rho-1}-F(0)^{2\rho-1}}{2\rho+1}= \dfrac{\|\chi'\|_\infty}{2\rho+1}\,,\\
|D_n|&\le \|\chi'\|_\infty^2 \int_{n}^{2n}F(x_2)^{2\rho+1}\ud x_2\le n\|\chi'\|_\infty^2\,,
\end{align*}
and using \eqref{egn} one has $\lim_{n\to+\infty} G_n=\rho(\rho-1)A$.
Hence choosing any value $\rho\in \big(\frac{1}{2}, 1\big)$ we have $G_n<0$ for large $n$, which
concludes the proof.

%
%
%
%
%
%
\subsection{Finiteness of the discrete spectrum}
In this section we prove that $Q_+$ has only finitely many eigenvalues in $(-\infty,\varepsilon_0)$.

We first introduce a pair of smooth functions $\chi_1,\chi_2:\mathbb{R} \rightarrow [0,\infty)$ such that $\chi_1(t)=1$ for $t \leq 1$, $\chi_2(t)=1$ for $t \geq 2$, and $\chi^2_1+\chi^2_2=1$.
We set, for $R >0$ and $j=1,2$, 
\[
\chi^{R}_j(x_1,x_2):=\chi_j\left(\frac{x_2-x_1}{R}\right)\, .
\]
Then, for any $\varphi \in \cD[Q_+]$ we have $\chi^{R}_j\varphi \in \cD[Q_+]$ and, by direct computation
\begin{align}
\label{EquationUSEFULProof}
Q_+[\varphi,\varphi]&=Q_+[\chi^{R}_1\varphi,\chi^{R}_1\varphi]
+Q_+[\chi^{R}_2\varphi,\chi^{R}_2\varphi]
-\iint_{\Omega_0}W_{R} \,\varphi^2\, \ud x_1 \ud x_2\,,\\
 \nonumber
W_{R}(x_1,x_2)&:=|\nabla \chi^{R}_1|^2 + |\nabla \chi^{R}_2|^2 =\frac{2}{R^2}\left[\chi^{\prime}_1\left(\frac{x_1-x_2}{R}\right)^2+\chi^{\prime}_2\left(\frac{x_2-x_1}{R}\right)^2 \right].
\end{align}
Consider two following (overlapping) subdomains of $\Omega_0$:
\begin{align*}
\Omega_1&:=\big\{(x_1,x_2)\in \Omega_0:\ x_2 < x_1+2R  \big\}\,,\\
\Omega_2&:=\big\{(x_1,x_2)\in \Omega_0:\ x_2 > x_1+R  \big\}\,,
\end{align*}
and define self-adjoint operators $Q_j$ in $L^2(\Omega_j)$, $j\in\{1,2\}$,
 by their forms
\begin{align*}
Q_1[\varphi,\varphi]&=\iint_{\Omega_1} \Big( \big| \nabla \varphi(x_1,x_2)\big|^2
+\big(v(x_1) -W_R(x_1,x_2)\big)\varphi(x_1,x_2)^2\Big)\ud x_1\ud x_2\,,\\
\cD[Q_1]&=\big\{\varphi \in H^1(\Omega_1):\, Q_1[\varphi,\varphi] < \infty\,,
\varphi=0 \text{ on the line } x_2=x_1+2R \big\},\\
Q_2[\varphi,\varphi]&=\iint_{\Omega_1} \Big( \big| \nabla \varphi(x_1,x_2)\big|^2
+\big(v(x_1) -W_R(x_1,x_2)\big)\varphi(x_1,x_2)^2\Big){ \ud x_1\ud x_2}\,,\\
\cD[Q_2]&=\big\{\varphi \in H^1(\Omega_2):\, Q_2[\varphi,\varphi] < \infty, \varphi=0
\text{ on the line } x_2=x_1+R\big\}\,.
\end{align*}

Let us return back to \eqref{EquationUSEFULProof}. The functions
$\chi^R_j \varphi$ vanish outside $\Omega_j$, $j\in\{1,2\}$, and
their restrictions to $\Omega_j$ belong to $\cD[Q_j]$.
In addition, one has $|\chi^R_1\varphi|^2+|\chi^R_1\varphi|^2= \varphi^2$ pointwise.
This allows one to rewrite \eqref{EquationUSEFULProof} as
\begin{equation}
   \label{eqq0}
Q_+[\varphi,\varphi]=Q_1[\chi^{R}_1\varphi,\chi^{R}_1\varphi]
+Q_2[\chi^{R}_2\varphi,\chi^{R}_2\varphi]\,.
\end{equation}
Consider an auxiliary operator $\widehat{Q}=Q_1\oplus Q_2$ defined on $L^2(\Omega_1)\oplus L^2(\Omega_2)$, then $\cD[\widehat Q]=\cD[Q_1]\times \cD[Q_2]$, with
\[
\widehat Q\big[(\varphi_1,\varphi_2),(\varphi_1,\varphi_2)\big]=Q_1[\varphi_1,\varphi_1]
+Q_2[\varphi_2,\varphi_2]\,.
\]
The linear map
\[
J:L^2(\Omega_0)\ni \varphi \mapsto (\chi^R_1\varphi,\chi_2^R\varphi)\in L^2(\Omega_1)\oplus L^2(\Omega_2)\,,
\]
is isometric and, hence, injective, with $J \cD[Q_+]\subset \cD[\widehat Q]$, and
Eq.~\eqref{eqq0} reads then as $Q_+[\varphi,\varphi]=\widehat Q[J\varphi,J\varphi]$. Hence, if one denotes be $E_n(L)$ the $n$th
eigenvalue of a self-adjoint operator $L$, then the min-max principle gives, for any $n\in\nz$,

\begin{align*}
%
E_n(Q_+)&=\inf\limits_{V_n\subset\cD[Q_+]}\sup\limits_{0\ne\varphi\in V_n}\frac{{Q}_+[\varphi,\varphi]}{\|\varphi\|^2_{L^2(\Omega_0)}}
=\inf\limits_{V_n\subset\cD[Q_+]}\sup\limits_{0\ne\varphi\in V_n}\frac{\widehat{Q}[J\varphi,J\varphi]}{\|J\varphi\|^2_{L^2(\Omega_1)\uplus L^2(\Omega_2)}}\\
&=\inf\limits_{U_n\subset J \cD[Q_+]}
\sup\limits_{0\ne\psi \in U_n}\frac{\widehat{Q}[\psi,\psi]}{\|\psi\|^2_{L^2(\Omega_1)\oplus L^2(\Omega_2)}}\\
&\ge
\inf\limits_{U_n\subset \cD[\widehat Q]}
\sup\limits_{0\ne\psi \in U_n}\frac{\widehat{Q}[\psi,\psi]}{\|\psi\|^2_{L^2(\Omega_1)\oplus L^2(\Omega_2)}}= { E_n(\widehat Q)\,,}
\end{align*}
%
where $V_n$ and $U_n$ stand for $n$-dimensional subspaces.
Hence, if for a self-adjoint operator $L$ and $\lambda\in\rz$ we denote
by $N(L,\lambda)$ the number of eigenvalues of $L$ in $(-\infty,\lambda)$,
then it follows from the above constructions that
\[
N(Q_+,\varepsilon_0)\le N(\widehat Q,\varepsilon_0)=N(Q_1,\varepsilon_0)+N(Q_2,\varepsilon_0).
\]
Hence, it is sufficient to show that $N(Q_j,\varepsilon_0)$ are finite for $j\in\{1,2\}$.
%
%
%
%

{ Let us start with $N(Q_1,\varepsilon_0)$:}
Consider the decomposition of $\Omega_1$ created by the line $x_1=L$, i.e.
\begin{align*}
\Omega_{1,\text{int}}&:=\big\{(x_1,x_2)\in \Omega_0:\ x_2 < x_1+2R\ \text{and}\ x_1 < R\big\}\,,\\
\Omega_{1,\text{ext}}&:=\big\{(x_1,x_2)\in \Omega_0:\ x_2 < x_1+2R\ \text{and}\ x_1 > R \big\}\,,
\end{align*}
and consider the operators $Q_{1,\bullet}$ in $L^2(Q_{1,\bullet})$ with
$\bullet\in\{\text{int},\text{ext}\}$, given by their forms
\begin{equation*}\begin{split}
&Q_{1,\bullet}[\varphi,\varphi]=\iint_{\Omega_{1,\bullet}}  \big| \nabla \varphi(x_1,x_2)\big|^2 \ud x_1\ud x_2\\
&\qquad \qquad \qquad+\iint_{\Omega_{1,\bullet}} \big(v(x_1) -W_R(x_1,x_2)\big)\varphi(x_1,x_2)^2\ud x_1\ud x_2\,,\\
&\cD[Q_{1,\bullet}]=\big\{\varphi \in H^1(\Omega_{1,\bullet}): \ Q_{1,\bullet}[\varphi,\varphi] < \infty \big\}\,.
\end{split}
\end{equation*}
The bilinear form for $Q_{1,\text{int}}\oplus Q_{1,\text{ext}}$ is an extension of the bilinear form for $Q_1$, and the min-max principle shows that the eigenvalues of $Q_1$
can not be lower than the respective eigenvalues of $Q_{1,\text{int}}\oplus Q_{1,\text{ext}}$.
In terms of the counting functions this leads to
\[
N(Q_1,\varepsilon_0)\le N(Q_{1,\text{int}}\oplus Q_{1,\text{ext}},\varepsilon_0)
=N(Q_{1,\text{int}},\varepsilon_0)+N(Q_{1,\text{ext}},\varepsilon_0)\,.
\]
The domain $\Omega_{1,\text{int}}$ is bounded, Lipschitz and $\cD[Q_{1,\text{int}}]\subset H^1(\Omega_{1,\text{int}})$ is compactly embedded into $L^2(\Omega_{1,\bullet})$, which implies that $Q_{1,\text{int}}$
is with compact resolvent, and then $N(Q_{1,\text{int}},\varepsilon_0)<\infty$ for any fixed $R>0$.
{ On the other hand,} the upper bound $\|W_R\|\le c/R^2$ with some $c>0$ and the assumption (C)
on the potential $v$ imply that for sufficiently large $R$ one has $v(x_1)-W_R(x_1,x_2)\ge\varepsilon_0$ for all $(x_1,x_2)\in \Omega_{1,\text{ext}}$. It follows that
$Q_{1,\text{ext}}$ has no spectrum below $\varepsilon_0$ and  $N(Q_{1,\text{ext}},\varepsilon_0)$. Therefore, there exists $R_0>0$ such that $N(Q_1,\varepsilon_0)<\infty$ for any $R>R_0$.

In order to conclude it remains to show that $N(Q_2,\varepsilon_0)<\infty$ for large $R>0${; note that $\Omega_2$ depends on $R$.}
Due to the fact that the functions in the form domain of $Q_2$ vanish at the line $x_2=x_1+R$
they can be extended by zero to functions in $H^1(\rz_+\times\rz)$.
Therefore, if one considers the operator $\widehat Q_2$ in $L^2(\rz_+\times\rz)$
given by
\begin{equation*}\begin{split}
&\widehat{Q}_2[\varphi,\varphi]:=\iint_{\rz_+ \times \rz}\,
\big| \nabla \varphi(x_1,x_2)\big|^2 \ud x_1\ud x_2 \\
& \qquad \qquad \qquad +\iint_{\rz_+ \times \rz}\big(v(x_1) -W_R(x_1,x_2)\big)\varphi(x_1,x_2)^2\ud x_1\ud x_2\,,\\
&\cD[\widehat{Q}_2]=\big\{\varphi \in H^1(\rz_+ \times \rz): \ \widehat{Q}_2[\varphi,\varphi] < \infty \big\}\,,
\end{split}
\end{equation*}
then it follows by the min-max principle that $N(Q_2,\varepsilon_0)\le N(\widehat Q_2,\varepsilon_0)$.
Therefore, it is sufficient to show that $N(\widehat Q_2,\varepsilon_0)<\infty$ for large $R$.

The subsequent construction is inspired by the representation
\[
\widehat Q_2=(h\otimes \eins + \eins\otimes q) -W_R\,,
\]
where $q$ is $f\mapsto -f''$ in $L^2(\rz)$ and $W_R$ is identified with the associated multplication operator.

Let $P$ be the { orthogonal projection} on $\rz\,\psi_0$ in $L^2(\rz_+)$, then
$\Pi:=P\otimes\eins$ is the { orthogonal projection} on $\psi_0 \otimes L^2(\rz)$
in $L^2(\rz_+\times\rz)$, i.e.
\begin{equation}\label{Profjection}
(\Pi \varphi)(x_1,x_2)=\psi_0(x_1) f(x_2),
\quad f(x_2):=\int_{\mathbb{R}_+} \varphi(x_1,x_2)\psi_0(x_1)\, \ud x_1{\,.}
\end{equation}
Notice that $\Pi$ is exactly the spectral projector on $\{\varepsilon_0\}$
for $h\otimes \eins$
(due to the fact that $\varepsilon_0$ is a simple eigenvalue, see Proposition~\ref{propa3})
and it commutes with $\eins\otimes q$.
We set $\Pi^\perp:=\eins-\Pi$. Taking into account that both
$\Pi\varphi$ and $\Pi^\perp\varphi$ are in $\cD[\widehat{Q}_2]$,
for $\varphi\in\cD[\widehat Q_2]$ we obtain
\begin{equation}\label{hq0}
\widehat{Q}_2[\varphi,\varphi]=\widehat{Q}_2[\Pi\varphi,\Pi\varphi]+\widehat{Q}_2[\Pi^\perp\varphi,\Pi^\perp\varphi]-2W_R[\Pi\varphi,\Pi^\perp\varphi]\,.
\end{equation}
As $W_R$ is bounded, using Cauchy-Schwarz and triangular inequalities we estimate
\begin{align*}
\big|2W_R[\Pi\varphi,\Pi^\perp\varphi]\big|&=
2\big|\langle W_R\Pi\varphi,\Pi^\perp\varphi\rangle_{L^2(\rz_+\times\rz)}\big|\\
&\leq 2 \|W_R\Pi \varphi\|_{L^2(\mathbb{R}_+ \times \mathbb{R})}\|\Pi^\perp \varphi\|_{L^2(\mathbb{R}_+ \times \mathbb{R})} \\
&\leq R\|W_R\Pi \varphi\|^2_{L^2(\rz_+ \times \rz)}+\frac{1}{R}\|\Pi^\perp \varphi\|^2_{L^2(\rz_+ \times \rz)}\, .
\end{align*}
Due to the assumption (B) on $v$, the eigenvalue $\varepsilon_0$ of $h$ is isolated, hence, $E_2:=\inf \big(\sigma(h)\setminus\{\varepsilon_0\}\big)>\varepsilon_0$, and
\[
(h\otimes \eins)[\Pi^\perp\varphi,\Pi^\perp\varphi]\ge E_2 \|\Pi^\perp \varphi\|^2_{L^2(\mathbb{R}_+ \times \mathbb{R})}\,.
\]
It follows that, taking into account that the operator $q$ is non-negative,
\begin{align*}
\widehat{Q}_2[\Pi^\perp \varphi,\Pi^\perp \varphi]
&=(h\otimes \eins)[\Pi^\perp\varphi,\Pi^\perp\varphi]
+(\eins\otimes q)[\Pi^\perp\varphi,\Pi^\perp\varphi]
-W_R[\Pi^\perp\varphi,\Pi^\perp\varphi]\\
&\ge E_2 \|\Pi^\perp \varphi\|^2_{L^2(\mathbb{R}_+ \times \mathbb{R})}
-W_R[\Pi^\perp\varphi,\Pi^\perp\varphi]\,.
\end{align*}
Summing up all the computations after \eqref{hq0} { yields}
\begin{multline}
  \label{ineq1}
\widehat{Q}_2[\varphi,\varphi]\ge 
\widehat{Q}_2[\Pi\varphi,\Pi\varphi]-R\|W_R\Pi \varphi\|^2_{L^2(\rz_+ \times \rz)}\\
+\Big(E_2 -\frac{1}{R}\Big)\|\Pi^\perp \varphi\|^2_{L^2(\mathbb{R}_+ \times \mathbb{R})}
-W_R[\Pi^\perp\varphi,\Pi^\perp\varphi]\,.
\end{multline}
Let $A$ be the self-adjoint operator in $\ran\Pi$ given by
\[
A[\Phi,\Phi]=\widehat{Q}_2[\Phi,\Phi]-R\|W_R\Phi\|^2_{L^2(\rz_+ \times \rz)}\,,
\quad
\cD[A]=\cD[\widehat{Q}_2]\cap\ran\Pi,
\]
and $B$ be the operator of multiplication by $E_2 -1/R-W_R$ in $\ran \Pi^\perp$, which is bounded and self-adjoint.
Considering the unitary map
\[
J:L^2(\rz_+\times\rz)\ni\varphi\mapsto (\Pi\varphi,\Pi^\perp \varphi)\in \ran\Pi \oplus \ran \Pi^\perp
\]
we rewrite \eqref{ineq1} as
$\widehat{Q}_2[\varphi,\varphi]\ge (A\oplus B)[J\varphi,J\varphi]$, which due to the min-max principle implies
\begin{equation}
  \label{eqnb}
N(\widehat{Q}_2,\varepsilon_0)\le N(A\oplus B,\varepsilon_0)=N(A,\varepsilon_0)+N(B,\varepsilon_0).
\end{equation}
As $E_2>\varepsilon_0$ is fixed and $\|W_R\|_\infty\le c/R^2$, for sufficiently large $R$ { and some $c>0$} one has the lower bound $E_2 -1/R+W_R\ge \varepsilon_0$ showing that
 $B$ has no spectrum in $(-\infty,\varepsilon_0)$ and { hence} $N(B,\varepsilon_0)=0$.
The estimate \eqref{eqnb} takes the form $N(\widehat{Q}_2,\varepsilon_0)\le N(A,\varepsilon_0)$, and now it is sufficient
to show that $N(A,\varepsilon_0)<\infty$ { for} $R$ being sufficiently large.

In order to study $A$ we rewrite, using the convention~\eqref{Profjection},
\begin{align*}
\widehat{Q}_2[\psi_0\otimes f,\psi_0\otimes f]&=\int_{\rz}
\Big(f^{\prime}(x_2)^2+\big(\varepsilon_0-U_R(x_2)\big)f(x_2)^2 \Big)\, \ud x_2\,, \\
U_R(x_2)&:=\int_{\mathbb{R}_+}W_R(x_1,x_2)\psi_{0}(x_2)^2\, \ud x_1\,, \nonumber\\
\big\|W_R \,(\psi_0\otimes f)\big\|^2_{L^2(\rz_+ \times \rz)}&=
\int_\rz V_R(x_2) f(x_2)^2 \ud x_2\,,\\
V_R(x_2)&:=\int_{\rz_+}W_R(x_1,x_2)^2\psi_{0}(x_2)^2\, \ud x_1\,.
\end{align*}
{ Recall} that for $\varphi \in \cD[\widehat{Q}_2]\equiv \cD[h\otimes \eins + \eins\otimes q]$
due to the spectral theorem one has $\Pi \varphi\equiv (P\otimes\eins)\varphi \equiv \psi_0\otimes f \in \cD[\eins\otimes q]$,
i.e. $f \in \cD[q]=H^1(\mathbb{R})$. { Consequently, one has}
\[
A[\psi_0\otimes f,\psi_0\otimes f] =\varepsilon_0\| f\|^2_{L^2(\rz)}+q_0[f,f]
\]
with $q_0$ being the self-adjoint operator in $L^2(\rz)$ given by the form
\[
q_0[f,f]:=\int_\rz\Big(
f'(x_2)^2 - Z_R(x_2) f(x_2)^2
\Big)\ud x_2\,,\quad Z_R=U_R+RV_R\,,
\]
defined on $\cD[q_0]=H^1(\rz)$. As the map $\ran\Pi\ni \psi_0\otimes f\mapsto f \in L^2(\rz)$ is unitary, one sees that
$A$ is unitarily equivalent to $q_0+\varepsilon$, which yields $N(A,\varepsilon_0)=N(q_0,0)$.

Now it is sufficient to show that $q_0$ has only finitely many negative eigenvalues. The task is simplified by the fact that $q_0$ is { a standard} one-dimensional
Schr\"odinger operator. { Recall that,} by construction one has
$W_R\in L^\infty$ and
\[
\text{supp}\,W_R\subset\big\{(x_1,x_2):R<x_2-x_1<2R\big\}\,,
\]
i.e. $W_R(x_1,x_2)$ vanishes except for $x_2-2R<x_1<x_2-R$.
Due to
\begin{equation}
   \label{zr1}
\begin{aligned}
Z_R(x_2)&=\int_{\rz_+\cap[x_2-2R,x_2-R]} \Big(W_R(x_1,x_2)+RW_R(x_1,x_2)^2\Big)\psi_{0}(x_1)^2\ud x_1\\
&\le \Big(\|W\|_\infty+R\|W\|_\infty^2\Big) \int_{\rz_+\cap[x_2-2R,x_2-R]} \psi_{0}(x_1)^2\ud x_1\\
&\le \|W\|_\infty+R\|W\|_\infty^2
\end{aligned}
\end{equation}
it follows that $Z_R$ is bounded, continuous, and $Z_R(x_2)=0$ for $x_2\le R$.
In view of the well-known Bargman estimate (see, e.g., Theorem 5.1 in Chapter 2.5 of~\cite{berezin1991schrodinger})
in order to obtain $N(q_0,0)<\infty$ it is sufficient to show
\begin{equation}
   \label{eq-bargm}
\int_\rz |x_2|\,Z_R(x_2)\, dx_2\equiv \int_R^\infty x_2\,Z_R(x_2)\, dx_2<\infty
\end{equation}
(recall that $W_R\ge 0$, and then $Z_R\ge 0$ as well).

In order to obtain \eqref{eq-bargm} we recall that due to the standard Agmon estimate (see e.g. Corollary~\ref{cora7} in Appendix) for some $a>0$ one has
\[
c_0:=\int_{\rz_+}e^{ax_1}\psi_0(x_1)^2\, \ud x_1<\infty\,.
\]
For $x_2\ge 2R$ one then estimates, using \eqref{zr1},
\begin{align*}
Z(x_2)
&\le \big(\|W\|_\infty+R\|W\|_\infty^2\big) \int_{x_2-2R}^{x_2-R} \psi^2_{0}(x_1)\, \ud x_1\\
&\le \big(\|W\|_\infty+R\|W\|_\infty^2\big) e^{-a(x_2-2R)}\int_{x_2-2R}^{x_2-R} e^{ax_1}\psi^2_{0}(x_1)\, \ud x_1\\
&\le c_0 \big(\|W\|_\infty+R\|W\|_\infty^2\big) e^{-a(x_2-2R)}=c_1 e^{-ax_2}
\end{align*}
with $c_1:=\Big(\|W\|_\infty+R\|W\|_\infty^2\Big) e^{2Ra}$.
Hence,
\begin{multline*}
\int_R^\infty x_2 Z_R(x_2)\, dx_2=\int_R^{2R} x_2 Z_R(x_2)\, dx_2+\int_{2R}^\infty x_2 Z_R(x_2)\, dx_2\\
\le 2R^2\|Z_R\|_\infty +c_1\int_R^\infty x_2 e^{-ax_2}\ud x_2<\infty\,.
\end{multline*}
This proves \eqref{eq-bargm} and completes the proof.

\appendix

%

\section{Some constructions for Schr\"odinger operators with singular potentials}\label{Appendix}
In this section we recall briefly some facts related to Schr\"odinger operators with singular potentials. All these facts are { well-known} to the specialists but we are not aware of their presentation within a single reference and in a suitable form under our rather weak assumptions on the potential $v$, and we decided to collect them here with proofs. An interested reader may refer e.g. to~\cite{Teschl:2013}
for a more detailed discussion of singular potentials.

For the whole of this section,
{ we write} $\rz_+=(0,\infty)$ and let $v\in L^1_\text{loc}(\rz_+)$ be a real-valued potential with $v_-:=\max\{-v,0\big\}\in L^\infty(\rz_+)$. Let $h$ be the self-adjoint operator in $L^2(\rz_+)$ generated by its bilinear form
\begin{gather*}
h[\varphi,\varphi]=\int_{\rz_+} \big(|\varphi'|^2 + v |\varphi|^2\big)\, \ud x\,,\\
\cD[h]=\big\{ \varphi\in H^1(\rz_+): \int_{\rz_+} v\,\varphi^2\, \ud x<\infty\big\}\,.
\end{gather*}
Recall that $\cD[h]$ stands for the form domain, while the operator domain is denoted { by }$\cD(h)$.
In other words, a function $\psi$ belongs to the operator domain $\cD(h)$ of $h$ and $h\psi=\psi_h$
if and only if
\[
\psi\in \cD[h] \quad \text{and} \quad
h[\varphi,\psi]=\int_{\rz_+} \varphi\, \psi_h\, \ud x
\quad \text{for all} \quad \varphi\in \cD[h]\,.
\]
As the preceding equality holds for all $\varphi\in C^\infty_0(\rz_+)\subset\cD[h]$, it follows
that $h$ acts as $h\psi=-\psi''+v\psi$.
We give a proof of the following technical fact:

\begin{prop}
\label{ertyu}
Let $\psi\in\cD(h)$ and $\chi\in C^\infty_0(\rz)$ with $\chi$ being constant in a neighborhood of $0$, then
$\chi \psi\in \cD(h)$ and $h(\chi \psi)=\chi h(\psi) -2\chi'\psi'-\chi'' \psi$.  
\end{prop}

\begin{proof}
Remark first that $\chi \psi\in\cD[h]$. { Then, we simply} need to show that
\begin{equation}
  \label{h01}
h[\varphi,\chi \psi]=\int_{\rz_+} \varphi (\chi h\psi -2\chi'\psi'-\chi'' \psi)\, \ud x
\end{equation}
for any $\varphi\in\cD[h]$. On the other hand, the assumption $\psi\in \cD(h)$ already gives
\begin{equation}
  \label{h02}
h[\chi\varphi,\psi]=\int_{\rz_+} \chi \varphi (-\psi''+v\psi)\, \ud x\,.
\end{equation}
Taking the difference between \eqref{h01} and \eqref{h02}
one sees that it is sufficient to show the equality
\[
h[\varphi,\chi \psi]-h[\chi\varphi,\psi]=-\int_{\rz_+} (2\varphi \chi'\psi'+\varphi\chi'' \psi)\, \ud x\,,
\]
which reads in a more detailed form as
\begin{equation}
    \label{aa1}
\int_{\rz_+} (\varphi'\chi'\psi-\varphi\chi'\psi')\, \ud x=-\int_{\rz_+} (2\varphi \chi'\psi'+\varphi\chi'' \psi)\, \ud x\,.
\end{equation}
One clearly has
\[
\begin{split}
\int_{\rz_+} (\varphi'\chi'\psi+\varphi \chi'\psi'+\varphi\chi'' \psi)\, \ud x&=\int_{\rz_+} (\varphi \chi'\psi)'\, \ud x\\
&=(\varphi \chi'\psi)(\infty)-(\varphi \chi'\psi)(0)=0\,.
\end{split}
\]
By regrouping the terms one arrives at \eqref{aa1}, which concludes the proof.
\end{proof}

For each $\psi\in\cD(h)$
one has $-\psi''+v\psi\in L^2(\rz_+)$. Due to the inclusions
$\cD(h)\subset\cD[h]\subset H^1(\rz_+)\subset L^\infty(\rz_+)$
it follows that $v\psi\in L^1_\text{loc}(\rz_+)$ and then
$\psi''\in L^1_\text{loc}(\rz_+)$ and $\psi'\in C^1(\rz_+)$.
That implies that the values $\psi(y)$ and $\psi'(y)$ make sense for any $y\in\rz_+$.
Let us add some precisions on the behavior near $0$ and $\infty$.

\begin{prop}\label{Randterme} Let $\psi\in\cD(h)$, then
\begin{equation}
\label{eqs6}
\lim_{x \rightarrow 0}(\psi^{\prime}\psi)(x)=:(\psi\psi^{\prime})(0)=0 \ , \quad \lim_{x \rightarrow \infty}(\psi^{\prime}\psi)(x)=:(\psi\psi^{\prime})(\infty)=0\, ,
\end{equation}
and the integration-by-parts formula
\begin{equation*}
\label{eqs6a}
\int_{0}^{y} \psi^{\prime}(x)\psi^{\prime}(x)\, \ud x=\psi(y)\psi^{\prime}(y)-\int_{0}^{y} \psi(x)\psi^{\prime \prime}(x)\, \ud x 
\end{equation*}
holds for $y\in\rz_+$.
\end{prop}
\begin{proof} In view of the above regularity of $\psi$, for any $0<\epsilon<y$ one has the standard integration by parts
\begin{equation}
   \label{eq1}
\int_\epsilon^y \psi'(x)^2 \, \ud x=(\psi \psi')(y)-(\psi \psi')(\epsilon)-\int_{\epsilon}^{y} \psi(x)\psi''(x)\, \ud x\,,
\end{equation}
and we need to show that the passage to the limit $\epsilon\to 0^+$ is possible.
By the definition of $\cD(h)$ one has
\begin{multline*}
\int_{\rz_+} \big(\psi'(x)^2 + v(x) \psi(x)^2\big)\, \ud x\equiv h[\psi,\psi]\equiv
\langle\psi,h\psi\rangle_{L^2(\rz_+)}\\
=\int_{\rz_+}\psi(-\psi''+v\psi)\, \ud x=\lim_{\epsilon\to 0^+} \int_{\epsilon}^{\epsilon^{-1}} \psi(-\psi''+v\psi)\, \ud x\\
=\lim_{\epsilon\to 0^+}
\Big\{\int_{\epsilon}^{\epsilon^{-1}} \big[\psi'(x)^2 + v(x) \psi(x)^2\big]\, \ud x
-(\psi \psi')(\epsilon^{-1})+(\psi \psi')(\epsilon)\Big\}
\end{multline*}
implying
\begin{equation}
    \label{eqa}
\lim_{\epsilon\to 0^+} \Big((\psi \psi')(\epsilon^{-1})-(\psi \psi')(\epsilon)\Big)=0\,.
\end{equation}
Let $\chi\in C^\infty_0(\rz)$ such that $\chi=1$ near zero, then
$\chi \psi\in\cD(h)$ due to Proposition~\ref{ertyu}, and \eqref{eqa}
also holds for $\psi$ replaced by $\chi \psi$.
As $\chi\psi$ is identically zero at infinity and coincides with $\psi$ near the origin,
one obtains
\[
\lim_{\epsilon\to 0^+} (\chi\psi)(\epsilon)(\chi \psi)'(\epsilon)\equiv \lim_{\epsilon\to 0^+} (\psi \psi')(\epsilon)=0\,.
\]
Using \eqref{eqa} again one has $\lim_{\epsilon\to 0^+} (\psi \psi')(\epsilon^{-1})\equiv \lim_{x\to +\infty} (\psi \psi')(x)=0$. By passing to the limit $\epsilon\to 0^+$ in \eqref{eq1} one concludes the proof.
\end{proof}

{ Assume from now on that the bottom $\varepsilon_0$ of the spectrum of $h$ is an eigenvalue.}

\begin{prop}\label{propa3}
The eigenvalue $\varepsilon_0$ is simple, and the corresponding eigenfunction $\psi_0$ can be chosen strictly positive.
\end{prop}

\begin{proof} { Let $\psi_0\in \ker(h-\varepsilon_0)$ with $\psi_0\not\equiv 0$ be given}. Due to the min-max principle { this} is equivalent to
\[
\dfrac{h[\psi_0,\psi_0]}{\|\psi_0\|_{L^2(\rz_+)}}=\min_{\psi\in\cD[h],\, \psi\ne 0}\dfrac{h[\psi,\psi]}{\|\psi\|_{L^2(\rz_+)}}\,.
\]
For $|\psi_0|\in \cD[h]$ one has $h\big[|\psi_0|,|\psi_0|\big]\le h[\psi_0,\psi_0]$
and $\big\||\psi_0|\big\|_{L^2(\rz_+)}=\|\psi_0\|_{L^2(\rz_+)}$,
which shows that $|\psi_0|\in \ker(h-\varepsilon_0)\subset\cD[h]\subset C^1(\rz_+)$.

Assume that $\psi_0(a)=0$ for some $a>0$, then from $|\psi_0|\in C^1(\rz_+)$ it follows that $\psi_0'(a)=0$.
Let us show that this implies $\psi_0(x)=0$ for all $x>0$. That is essentially
Gronwall's lemma, but we prefer to include it for completeness.
To be definite, consider $x>a$ (the other case $x<a$ is considered in the same way).
The fact $h\psi_0=\varepsilon_0\psi_0$ can be rewritten as
\[
\Psi(x)=\int_a^x M(t) \Psi(t)\, \ud t\,, \quad
\Psi(x)=\begin{pmatrix}
\psi_0(x)\\ \psi_0'(x)
\end{pmatrix}\,,
\quad
M(x)=\begin{pmatrix}
0 & 1\\
v(x)-E & 0
\end{pmatrix}\, { .}
\]
{ Then} for $f:=|\Psi|_{\rz^2}\ge 0$ and $m:=\|M\|\in L^1_\text{loc}(\rz_+)$ one has
\[
f(x)\le\int_a^x m(t) f(t)\, \ud t\,
\le \varepsilon +\int_a^x m(t) f(t)\, \ud t=:\Phi(x){ Punkt weg}
\]
for all $\varepsilon>0$ and $x>a$. Therefore, $\Phi'(x)/\Phi(x)\le m(x)$,
so by integrating between $a$ and $x$ one arrives at
\[
\Phi(x)\le \Phi(a)\exp\int_a^x m(t)\, \ud t \text{ for $x>a$\,.}
\]
Due to $f\le\Phi$ and $\Phi(a)=\varepsilon$ one obtains
\[
0\le f (x)\le \varepsilon \exp\int_a^x m(t)\, \ud t\,, \quad x>a\,.
\]
As $\varepsilon>0$ is arbitrary, one obtains $f(x)=0$ for $x>a$, which implies $\psi_0(x)=0$ for $x>a$.

We conclude that an eigenfunction $\psi_0\in\ker(h-\varepsilon_0)$ cannot vanish, hence, up to a multiplicative
factor it is strictly positive. As two strictly positive functions cannot be orthogonal in $L^2$,
the eigenvalue $\varepsilon_0$ is simple.
\end{proof}

For the rest of the section, let $\psi_0$ be the strictly positive eigenfunction for $\varepsilon_0$, with a unit $L^2$-norm.

\begin{prop}\label{RegularityFunction} Let $\rho >\dfrac{1}{2}$, then the function
	\begin{equation*}
\phi:y\mapsto \left(\int_{0}^{y}\psi_0(x)^2\, \ud x\right)^{\rho}
	\end{equation*}
	is in $H^1(0,a)$ for any $a>0$.
\end{prop}
\begin{proof} Since $\phi \in L^{\infty}(\rz_+)$, we only have to take care of the derivative. A direct calculation shows that
	\begin{equation*}
	\int_{0}^{a} \phi^{\prime}(y)^2\, \ud y=\rho^2\int_{0}^{a}
	\psi_0(y)^4\left(\displaystyle\int_{0}^{y}\psi_0(x)^2\, \ud x\right)^{2(\rho-1)}\, \ud y\, .
	\end{equation*}
	Let $y(\cdot)$ be the inverse of
	\[
	F:=y\mapsto \int_{0}^{y}\psi_0(x)^2\, \ud x\,,
	\]
	which is a diffeomorphism due to $\psi_0>0$ (Proposition~\ref{propa3}), then
	\begin{equation*}
	y^{\prime}(s)=\frac{1}{F^{\prime}(F^{-1}(s))}=\frac{1}{\psi_0\big(F^{-1}(s)\big)^2}\,,
	\end{equation*}
	and consequently
	\[
	\|\phi^{\prime}\|^2_{L^2(0,a)}= \rho^2\int_{0}^{F(a)} \psi_0\big(F^{-1}(s)\big)^2 s^{2(\rho-1)} \, \ud s\,.
	\]
	As $\psi_0\in L^\infty(\rz_+)$, the integral is finite for $\rho > \frac{1}{2}$.
\end{proof}

For the rest of the section we assume finally that
\begin{equation}
    \label{vvv}
v_\infty:=\liminf_{x\to+\infty} v(x)>\varepsilon_0\,.
\end{equation}
For $L>0$, define two operators $h^{N/D}_L$ in $L^2(0,L)$ by
\[
h^{N/D}_L[\varphi,\varphi]=\int_0^L \big((\varphi')^2 + v \varphi^2\big)\, \ud x
\]
with form domains
\begin{align*}
\cD[h^N_L]&=\Big\{\varphi\in H^1(0,L):\, \int_0^L v\varphi^2\, \ud x <\infty\Big\}\,,\\
\cD[h^D_L]&=\big\{\varphi\in \cD[h^N_L]:\,\varphi(L)=0\big\}\,,
\end{align*}
and denote by $\varepsilon^{(N/D)}_0(L)$ the respective lowest
eigenvalues.  
\begin{prop}\label{CorollaryConvergenceEnergy} There holds $\lim_{L \rightarrow \infty}\varepsilon^{(N)}_0(L)=\varepsilon_0$.
\end{prop}
\begin{proof}
By the min-max principle one has $\varepsilon_0 \leq \varepsilon^{(D)}_0(L)$ for all $L > 0$.

Let $\widetilde{h}^{N}_L$ be the operator on $L^{2}(L,\infty)$ associated with the form
\begin{gather*}
\widetilde{h}_{L}[\varphi,\varphi]:=\int_{L}^{\infty}\big((\varphi')^2+v \varphi^2\big) \, \ud x,\\
\cD[\widetilde{h}_{L}]:=\{\varphi\in H^1(L,\infty):\ \widetilde{h}_{L}[\varphi,\varphi] < \infty  \}\,.
\end{gather*}
Again, the min-max principle then implies that
\begin{equation}
   \label{eps00}
\varepsilon_0=\equiv \inf\sigma(h) \geq \inf \sigma(h^{N}_L \oplus \tilde{h}^{N}_L)\,.
\end{equation}
On the other hand,
$\inf \sigma(h^{N}_L \oplus \widetilde{h}^{N}_L)=\min\big\{ \varepsilon_0^N(L), \inf\sigma(\widetilde{h}^{N}_L)\big\}$, and due to the assumption~\eqref{vvv}
for sufficiently large $L$ one has $\inf\sigma(\widetilde{h}^{N}_L)>\varepsilon_0$.
	It follows from \eqref{eps00} that $\varepsilon_0\ge\varepsilon^{(N)}_0(L)$ for large $L$.
	
	Now let us take  $\chi_1,\chi_2 \in C^{\infty}(\mathbb{R})$ with
	\[
	\text{$\chi^2_1+\chi^2_2=1$, that $\chi_1(t)=1$ for $t\leq\frac{1}{2}$ and $\chi_1(t)=0$ for $t\ge 1$\,,}
	\]
	and set $\chi^{L}_j(t):=\chi_j(t/L)$, $j=1,2$. For any $\varphi \in \cD[h^N_L]$ we obtain
	\begin{equation*}\begin{split}
	h^{N}_L[\varphi,\varphi]&=h^{N}_L[\chi^{L}_1\varphi,\chi^{L}_1\varphi]+h^{N}_L[\chi^{L}_2\varphi,\chi^{L}_2\varphi]-\int_{0}^{L}\left[((\chi^{L}_1)^{\prime})^2+ ((\chi^{L}_2)^{\prime})^2\right]
	\varphi^2\, \ud x \\
	&\geq h^{N}_L[\chi^{L}_1\varphi,\chi^{L}_1\varphi]+h^{N}_L[\chi^{L}_2\varphi,\chi^{L}_2\varphi]-\frac{c}{L^2}\,\|\varphi\|^2_{L^2(0,L)}\, 
	\end{split}
	\end{equation*}
	for some constant $c > 0$. Since $(\chi^{L}_1\varphi)(x)=0$ for $x\ge L$ we conclude that
	$\chi^{L}_1\varphi\in \cD[h^D_L]$ and then
	\[
	h^{N}_L[\chi^{L}_1\varphi,\chi^{L}_1\varphi]=h^{D}_L[\chi^{L}_1\varphi,\chi^{L}_1\varphi]\\
	\ge\varepsilon^{(D)}_0(L) \|\chi^{L}_1\varphi\|^2_{L^2(0,L)}\ge
		\varepsilon_0 \|\chi^{L}_1\varphi\|^2_{L^2(0,L)}\,.
	\]
	Using the assumption \eqref{vvv} on $v$, one can choose $L$ sufficiently large to have
	$v\ge \varepsilon_0$ in $(-L/2,L)$. Due to { $\supp \chi^{L}_2\varphi\subset[L/2,L]$}
	there holds
	\begin{multline*}
	h^{N}_L[\chi^{L}_2\varphi,\chi^{L}_2\varphi]\ge \int_{L/2}^L v \,(\chi^{L}_1\varphi)^2\ud x\ge
	\varepsilon_0 \int_{L/2}^L (\chi^{L}_1\varphi)^2\ud x
	=\varepsilon_0 \|\chi^{L}_2\varphi\|^2_{L^2(0,L)}\,.
	\end{multline*}
Therefore, for large $L$ one has, uniformly in $\varphi\in \cD[h^N_L]$,
\begin{align*}
h^{N}_L[\varphi,\varphi]&\geq
\varepsilon_0 \|\chi^{L}_1\varphi\|^2_{L^2(0,L)}
+\varepsilon_0 \|\chi^{L}_2\varphi\|^2_{L^2(0,L)}-\frac{c}{L^2}\|\varphi\|^2_{L^2(0,L)}\\
&=\Big(\varepsilon_0-\frac{c}{L^2}\Big)\|\varphi\|^2_{L^2(0,L)}\,,
\end{align*}
{ which implies $\varepsilon^{(N)}_0(L)\geq
\varepsilon_0(L)-c/L^2$ due to the min-max principle}. Summing up we obtain, for $L>0$ large enough,
	\begin{equation*}
	\varepsilon_0-\frac{c}{L^2} \leq \varepsilon^{(N)}_0(L) \leq \varepsilon_0\,,
	\end{equation*}
	which proves the statement.
\end{proof}
In the next result we recall an Agmon-type estimate for the ground state $\psi_0$ of $h$. Recall that $\psi_0$ was chosen strictly positive and normalized in $L^2(\rz_+)$. 
\begin{prop}[Agmon-type estimate]\label{AgmonEstimate}
For any $\theta\in(0,1)$ there is $R>0$ with $v(x)\ge \varepsilon_0$ for $x\ge R$ such that
	\begin{equation*}
	\int_{0}^{\infty} e^{2\theta\Phi(x)}\psi_0(x)^2\, \ud x < \infty\, , 
	\qquad
	\Phi(x):=\begin{cases}
	0\,,& x \leq R\, , \\
	\displaystyle\int_{R}^{x}\sqrt{v(t)-\varepsilon_0}\, \ud t\,,& x > R\, .
	\end{cases}
	\end{equation*}
\end{prop}
\begin{proof}
Let us take a sufficiently large $R>0$ such that $v(x)\ge \varepsilon_0$ for $x\ge R$; the value of $R$ will be adjusted later.
Define $\Phi$ as above, and for $L > 0$ define
\[
\phi_{L}(x):=\theta \min\{\Phi(x),L\}\,,
\]
then $\phi_{L} \in L^{\infty}(\mathbb{R}_+)$ and $\big|\phi^{\prime}_{L}(x)\big| \leq \theta \eins_{x> R}(x)\sqrt{v(x)-\varepsilon_0}$,
where $\eins_{x> R}$ stands for the indicator function of the set $\big\{x \in \mathbb{R}_+:x> R \big\}$.

Let us show first that
\begin{equation}
   \label{agm1}
\text{for any}\quad c\in\rz \quad \text{one has} \quad e^{c\phi_{L}}\psi_0\in \cD[h]\,.
\end{equation}
By construction, $e^{c\phi_{L}}\in L^\infty(\rz_+)$, so $e^{c\phi_{L}}\psi_0\in L^2(\rz_+)$ and
\[
\int_{\rz_+}v (e^{c\phi_{L}}\psi_0)^2\, \ud x<\infty \text{ due to } \int_{\rz_+}v \psi_0^2\, \ud x<\infty\,.
\]
Furthermore, $(e^{c\phi_{L}}\psi_0)'=c\phi'_L e^{c\phi_{L}}\psi_0+ e^{c\phi_{L}}\psi_0$,
and the second summand is in $L^2(\rz_+)$ due to $\psi_0\in H^1(\rz_+)$, while the first summand is finite
due to
\[
\int_{\rz_+}(\phi'_L e^{c\phi_{L}}\psi_0)^2\, \ud x\le \theta^2 e^{2cL} \int_R^\infty (v-\varepsilon_0)\psi_0^2\, \ud x<\infty\,.
\]
Hence, the claim \eqref{agm1} is proved.

Now we compute
\begin{multline*}
	h[e^{\phi_L}\psi_0,e^{\phi_L}\psi_0]=\int_{\rz_+} \Big( \big(\phi'_L e^{\phi_L}\psi_0+e^{\phi_L}\psi'_0\big)^2 +v (e^{\phi_L}\psi_0)^2\Big)\ud x\\
    =\int_{\rz_+}(\phi_{L}^{\prime})^2e^{2\phi_{L}}\psi_0^2\, \ud x + \int_{\rz_+}\Big((e^{2\phi_L}\psi_0)' \psi_0' +v \,e^{2\phi_L}\psi_0\, \psi_0\, \Big)\,\ud x\,.
\end{multline*}
Due to \eqref{agm1} one can transform the last summand on the right-hand side as
\begin{multline*}
\int_{\rz_+}\Big((e^{2\phi_L}\psi_0)' \psi_0' +v \,e^{2\phi_L}\psi_0\, \psi_0\,\Big)\, \ud x=h\big[e^{2\phi_L}\psi_0,\psi_0\big]\\
=\langle e^{2\phi_L}\psi_0, h\psi_0\rangle_{L^2(\rz_+)}=\varepsilon_0 \big\langle e^{2\phi_L}\psi_0, \psi_0\big\rangle_{L^2(\rz_+)}=
\varepsilon_0\int_{\rz_+} e^{2\phi_L}\psi_0^2\ud x\,,
\end{multline*}
which yields
\begin{equation}
  \label{q01}
h[e^{\phi_L}\psi_0,e^{\phi_L}\psi_0]=\int_{\rz_+}\Big( (\phi_{L}^{\prime})^2 + \varepsilon_0 \Big) e^{2\phi_L}\psi_0^2\,\ud x\,.
\end{equation}
Now let us pick any $\delta>0$. The min-max principle applied to $h^N_R$ gives
\begin{equation*}
\int_0^R\Big( \big((e^{\phi_{L}}\psi_0)'\big)^2+v(e^{\phi_{L}}\psi_0)^2 \Big)\, \ud x \geq \varepsilon^{(N)}_0(R)\int_{0}^R e^{2\phi_{L}}\psi_0^2 \, \ud x\,{ .}
\end{equation*}
{ Hence, for large $R > 0$ one has  $\varepsilon^{(N)}_0(R)\ge\varepsilon_0-\delta$} due to Proposition~\ref{CorollaryConvergenceEnergy},
and
\begin{multline*}
h[e^{\phi_{L}}\psi_0,e^{\phi_{L}}\psi_0]=\int_{\mathbb{R}_+}\Big( \big((e^{\phi_{L}}\psi_0)'\big)^2+v(e^{\phi_{L}}\psi_0)^2 \Big) \, \ud x\\
=\int_0^R \Big( \big((e^{\phi_{L}}\psi_0)'\big)^2+v(e^{\phi_{L}}\psi_0)^2 \Big)\, \ud x
+ \int_R^\infty\Big( \big(e^{\phi_{L}}\psi_0)'\big)^2+v(e^{\phi_{L}}\psi_0)^2 \Big) \, \ud x\\
\ge \varepsilon^{(N)}_0(R) \int_{0}^{R}e^{2\phi_{L}}\psi_0^2 \, \ud x + \int_R^\infty ve^{2\phi_{L}}\psi_0^2 \, \ud x\\
\ge (\varepsilon_0-\delta)\int_{0}^{R}e^{2\phi_{L}}\psi_0^2 \, \ud x + \int_R^\infty ve^{2\phi_{L}}\psi_0^2 \, \ud x\,.
\end{multline*}
By combining this last inequality with \eqref{q01} we arrive at
\begin{equation*}
\int_{\mathbb{R}_+}\left[(\phi_{L}^{\prime})^2+\varepsilon_0\right]e^{2\phi_{L}}\psi_0^2 \, \ud x
\geq (\varepsilon_0-\delta)\int_{0}^{R}e^{2\phi_{L}}\psi_0^2 \, \ud x +\int_{R}^{\infty}v e^{2\phi_{L}}\psi_0^2 \, \ud x\,.
\end{equation*}
This rewrites as
\[
\int_{0}^{R} \big((\phi_{L}^{\prime})^2+\delta)e^{2\phi_{L}}\psi_0^2 \, \ud x
\ge
\int_{R}^{\infty} \big(v-\varepsilon_0-(\phi_{L}^{\prime})^2\big) e^{2\phi_{L}}\psi_0^2 \, \ud x\,
\]
and taking into account the above choice of $R$ and $\phi_L$ we arrive at
\begin{align*}
\delta\int_{0}^{R} \psi_0^2 \, \ud x &\ge
\int_{R}^{\infty} \big(v-\varepsilon_0-(\phi_{L}^{\prime})^2\big) e^{2\phi_{L}}\psi_0^2 \, \ud x\\
&\ge (1-\theta^2)\int_{R}^{\infty} (v-\varepsilon_0) e^{2\phi_{L}}\psi_0^2 \, \ud x\,.
\end{align*}
As $\delta>0$ was arbitrary, we may assume that $\delta<v_\infty-\varepsilon_0$, then for large $R$ one has $v-\varepsilon_0\ge \delta$ in $(R,\infty)$,
and it follows from the preceding inequality that
\[
\int_{0}^{R} \psi_0^2 \, \ud x\ge (1-\theta^2)\int_{R}^{\infty}  e^{2\phi_{L}}\psi_0^2 \, \ud x\,.
\]
Consequently,
\begin{multline*}
\int_{0}^{\infty}e^{2\phi_{L}}\psi_0^2\, \ud x=\int_{0}^{R}e^{2\phi_{L}}\psi_0^2\, \ud x+\int_{R}^{\infty}e^{2\phi_{L}}\psi_0^2\, \ud x\\
=\int_{0}^{R}\psi_0^2\, \ud x+\int_{R}^{\infty}e^{2\phi_{L}}\psi_0^2\, \ud x \le \big(1+(1-\theta^2)\big)\int_{0}^{R}\psi_0^2\, \ud x\le  2-\theta^2,
\end{multline*}
or, in a detailed form,
\[
\int_{0}^{\infty}\exp \Big( 2\theta \min\big\{\Phi(x),L\big\}\Big)\psi_0^2\, \ud x
\le 2-\theta^2.
\]
As the constant on the right-hand side is independent of the choice of $L$, the statement then follows by taking the limit $L \rightarrow \infty$.
\end{proof}

We prefer to give a simplified version of the preceding estimate, which will be easier to use in the main text:
\begin{cor}\label{cora7}
For some $a>0$ there holds
\[
\int_{\rz_+} e^{ax} \psi_0(x)^2\, \ud x<\infty\,.
\]
\end{cor}

\begin{proof}
Due to the assumption \eqref{vvv} on $v$, for some $b>0$ one has $v(x)-\varepsilon_0\ge b$
for large $x$, and then the function $\Phi$ in Proposition~\ref{AgmonEstimate} satisfies the inequality
$\Phi(x)\ge \sqrt{b}\,(x-R)-c$ for all $x$ (with a fixed $c>0$), which leads to
\begin{align*}
\int_{\rz_+} e^{2\theta \sqrt{b}\, x}\psi_0^2\ud x&= e^{2\theta \sqrt{b}\, R+2\theta c} \int_{\rz_+} e^{2\theta \sqrt{b}\, (x-R)-2\theta c}\psi_0^2\ud x\\
&\le e^{2\theta \sqrt{b}\, R+2\theta c}\int_{\rz_+} e^{2\theta \Phi} \psi_0^2\,\ud x<\infty\,,
\end{align*}
which gives the claim with $a:=2\theta \sqrt{b}$.
\end{proof}

We finish this appendix by mentioning two classical cases for which the assumption \eqref{vvv} is satisfied.
Recall that 
\[
v_\infty:=\liminf_{x\to +\infty} v(x)\,.
\]

\begin{prop}\label{vinf} There holds $\inf\sigma_\mathrm{ess}(h)\ge v_\infty$.
\end{prop}

\begin{proof}
Let { $\widetilde{h}_{L}$ }be the operator on $L^{2}(L,\infty)$ given by its bilinear form
\begin{gather*}
\widetilde{h}_{L}[\varphi,\varphi]:=\int_{L}^{\infty}\big((\varphi')^2+v \varphi^2\big) \, \ud x\,,\\
\cD[\widetilde{h}_{L}]:=\big\{\varphi\in H^1(L,\infty):\ \widetilde{h}_{L}[\varphi,\varphi] < \infty  \big\}\,.
\end{gather*}
{ Then the min-max principle }implies $\inf\sigma_\mathrm{ess}(h)\ge \inf\sigma_\mathrm{ess}(h^N_L \oplus \widetilde{h}_{L})$ for any $L>0$.
The operator $h^N_L$ has compact resolvent and an empty essential spectrum, hence, $\inf\sigma_\mathrm{ess}(h)\ge \inf\sigma_\mathrm{ess}(\widetilde{h}_{L})$.
For any $a<v_\infty$ one can choose a large $L>0$ to have $v\ge a$ in $(L,\infty)$, which leads to $\inf\sigma(\widetilde h_L)\ge a$.
It follows that
\[
\inf\sigma_\mathrm{ess}(h)\ge \inf\sigma_\mathrm{ess}(\widetilde{h}_{L})\ge \inf\sigma(\widetilde h_L)\ge a\,.
\]
As $a<v_\infty$ is arbitrary, this gives the result.
\end{proof}

\begin{prop}\label{cond1}
If $v_\infty=+\infty$, then the bottom of the spectrum of $h$ is an
isolated eigenvalue $\varepsilon_0$ with $\varepsilon_0<v_\infty$.
\end{prop}
\begin{proof}
In this case $\sigma_\text{ess}(h)=\emptyset$ by Proposition~\ref{vinf},
i.e. $h$ is with compact resolvent. Its lowest eigenvalue $\varepsilon_0$
is then automatically
isolated, and the inequality $\varepsilon_0<v_\infty$ is just the finiteness of $\varepsilon_0$.
\end{proof}

\begin{prop}\label{cond2}
Assume that $v_\infty<+\infty$ and that $v-v_\infty\in L^1(\rz_+)$ with
\[
\int_{\rz_+} \big(v(x)-v_\infty\big)\ud x<0\,,
\]
then the bottom $\varepsilon_0$ of the spectrum of $h$ is an isolated eigenvalue, and it satisfies
$\varepsilon_0<v_\infty$.
\end{prop}
\begin{proof}
In view of Proposition~\ref{vinf} it is sufficient to establish
the existence of eigenvalues in $(-\infty,v_\infty)$,
for which it is sufficient to find a function $\varphi\in \cD[h]$ with
$h[\varphi,\varphi]-v_\infty \|\varphi\|^2_{L^2(\rz_+)}<0$.

For $\delta>0$ consider $\varphi:x\mapsto e^{-\delta x}$, then $\varphi\in \cD[h]$ with
\begin{align*}
h[\varphi,\varphi]-v_\infty \|\varphi\|^2_{L^2(\rz_+)}&=\delta^2\int_{\rz_+} e^{-2\delta x}\,\ud x
+\int_{\rz_+} \big(v(x)-v_\infty\big)e^{-2\delta x}\ud x\\
&=\dfrac{\delta}{2} +\int_{\rz_+} \big(v(x)-v_\infty\big)e^{-2\delta x}\ud x\,,
\end{align*}
and the right-hand side converges to a strictly negative limit as $\delta\to 0^+$.
\end{proof}

{
\begin{remark}\label{lastrem}
It is easily seen that all assertions of this Appendix, except Proposition~\ref{cond2},
remain valid for if one replaces the operator $h$
by the operator $h_0$ defined in \eqref{eqh0}, which provides
necessary technical components to prove Proposition~\ref{prop23}.
\end{remark}
}

%
%

\begin{thebibliography}{EGNT}

\bibitem[BCS]{BCSI}
J.~Bardeen, L.~N. Cooper, and J.~R. Schrieffer, \emph{Theory of
  superconductivity.} Phys. Rev. \textbf{108} (1957), 1175--1204.

\bibitem[BERW]{busch}
T. Busch, B.-G. Englert, K. Rzazewski,
M. Wilkens, \emph{Two cold atoms in a harmonic trap.}
Found. Phys. {\bf 28}:4 (1998) 549--559.

\bibitem[BS]{berezin1991schrodinger}
F.~A. Berezin and M.~A. Shubin, \emph{The {S}chr{\"o}dinger equation}, Kluwer
  Academic, 1991.

\bibitem[C]{CooperBoundElectron}
L.~N. Cooper, \emph{Bound electron pairs in a degenerate {F}ermi gas}, Phys.
  Rev. \textbf{104} (1956), 1189--1190.


\bibitem[EGNT]{Teschl:2013}
J.~Eckhardt, F.~Gesztesy, R.~Nichols, and G.~Teschl, \emph{Weyl-{T}itchmarsh
  theory for {S}turm-{L}iouville operators with distributional potentials},
  Opuscula Math. \textbf{33} (2013), 467--563.


\bibitem[HM]{hm} B. Helffer and A. Morame, \emph{Magnetic bottles for the Neumann problem: The case
of dimension~3.} Proc. Indian Acad. Sci. (Math. Sci.) {\bf 112}:1 (2002), 71--84.


\bibitem[K17]{KernerElectronPairs}
J.~Kerner, \emph{{On bound electron pairs in a quantum wire}}. Preprint arXiv:1708.03753.

\bibitem[K18]{KernerInteractingPairs}
J.~Kerner, \emph{On pairs of interacting electrons in a quantum wire}, J.~Math. Phys.
\textbf{59} (2018), no.~6, 063504.

\bibitem[KM]{KernerMuhlenbruch2}
J.~Kerner and T.~M\"{u}hlenbruch, \emph{On a two-particle bound system on the
  half-line}, Rep. Math. Phys. \textbf{80} (2017), no.~2, 143 --
  151.

\bibitem[KP]{KP}
M. Khalile and K. Pankrashkin, \emph{Eigenvalues of Robin Laplacians in infinite sectors.}
Math. Nachr. {\bf 291} (2018) 928--965.
	
\bibitem[LP]{lupan} K. Lu and X.-B. Pan, \emph{Surface nucleation of superconductivity
in 3-dimensions.} J.~Differential Equ. {\bf 168}:2 (2000), 386--452.
	
\bibitem[MT]{MT}
A. Morame and F. Truc, \emph{Remarks on the spectrum of the Neumann problem with magnetic field in the half-space.}
J.~Math. Phys. {\bf 46} (2005) 1--13.

\bibitem[P]{P}
K. Pankrashkin, \emph{Variational proof of the existence of eigenvalues for star graphs.}
In J. Dittrich, H. Kova\v{r}\'{\i}k, A. Laptev (Eds.): \emph{Functional Analysis and Operator Theory for Quantum Physics. Pavel Exner Anniversary Volume} (EMS Series of Congress Reports, vol. 12, 2017)
447--458.


\end{thebibliography}
%
\end{document}